\documentclass[a4paper, amsfonts, amssymb, amsmath, reprint, showkeys, nofootinbib, superscriptaddress]{revtex4-1} 
\usepackage[utf8]{inputenc}
\usepackage[T1]{fontenc}
\usepackage{amssymb}
\usepackage[colorinlistoftodos, color=green!40, prependcaption]{todonotes}
\usepackage{amsthm}
\usepackage{mathtools}
\usepackage{physics}
\usepackage{diagbox}
\usepackage{makecell}
\usepackage[left=15mm,right=15mm,top=30mm, columnsep=15pt]{geometry}
\usepackage{adjustbox}
\usepackage{relsize}
\usepackage{placeins}
\usepackage{hyperref}
\usepackage{extarrows}
\usepackage{csquotes}

\usepackage[]{graphicx}
\usepackage{float}
\usepackage{qcircuit}
\usepackage{amsmath}
\usepackage{braket}
\usepackage{multirow}
\usepackage{verbatim}
\usepackage{dsfont}
\usepackage{bm}
\usepackage[shortlabels]{enumitem}
\usepackage{multirow}
\usepackage{rotating}
\numberwithin{equation}{section}

\newtheorem{lemma}{Lemma}
\newtheorem{theorem}{Theorem}

\begin{document}
\title{Experimentally feasible computational advantage from quantum superposition of gate orders}

\author{Martin J. Renner,}
    \email{martin.renner@univie.ac.at}
    \affiliation{University of Vienna, Faculty of Physics, Vienna Center for Quantum Science and Technology (VCQ), Boltzmanngasse 5, 1090 Vienna, Austria}
    \affiliation{Institute for Quantum Optics and Quantum Information (IQOQI), Austrian Academy of Sciences, Boltzmanngasse 3, 1090 Vienna, Austria}

\author{\v{C}aslav Brukner}
    \affiliation{University of Vienna, Faculty of Physics, Vienna Center for Quantum Science and Technology (VCQ), Boltzmanngasse 5, 1090 Vienna, Austria}
    \affiliation{Institute for Quantum Optics and Quantum Information (IQOQI), Austrian Academy of Sciences, Boltzmanngasse 3, 1090 Vienna, Austria}

\date{\today}

\begin{abstract}
In an ordinary quantum algorithm the gates are applied in a fixed order on the systems. The introduction of indefinite causal structures allows to relax this constraint and control the order of the gates with an additional quantum state. It is known that this quantum-controlled ordering of gates can reduce the query complexity in deciding a property of black-box unitaries with respect to the best algorithm in which the gates are applied in a fixed order. However, all tasks explicitly found so far require unitaries that either act on unbounded dimensional quantum systems in the asymptotic limit (the limiting case of a large number of black-box gates) or act on qubits, but then involve only a few unitaries. Here we introduce tasks (1) for which there is a provable computational advantage of a quantum-controlled ordering of gates in the asymptotic case and (2) that require only qubit gates and are therefore suitable to demonstrate this advantage experimentally. We study their solutions with the quantum-$n$-switch and within the quantum circuit model and find that while the $n$-switch requires to call each gate only once, a causal algorithm has to call at least $2n-1$ gates. Furthermore, the best known solution with a fixed gate ordering calls $O(n\log_2{(n)})$ gates.
\end{abstract}

\maketitle

\section{Introduction}
Causality is one of the most fundamental concepts in science and deeply embedded in the concept of computation. In ordinary quantum algorithms, represented within the quantum circuit model, the gates act in a fixed order on the systems. However, the study of causality at the intersection between quantum mechanics and gravity within the last two decades \cite{hardy2005probability, Zych_2019} suggested that quantum computation can be extended to more general scenarios, in which the order of the gates is controlled with an additional quantum state \cite{Chiribella_2013, Oreshkov_2012}. The use of indefinite causal structures provide numerous advantages in the field of quantum information. For instance, they lead to an exponential reduction for certain communication tasks \cite{Guerin_2016} and offer advantages in channel discrimination tasks \cite{bavaresco2020strict}. Moreover, they allow to transfer information through zero-capacity channels \cite{PhysRevLett.120.120502, salek2018quantum, chiribella2018indefinite, Guo2020, goswami2020}, although the same effect appears in causal circuits \cite{abbott2018communication, Gu_rin_2019, Rubino_2021}. Beside the theoretical interest of indefinite causal structures, including the study of the computational complexity \cite{Ara_jo_2017, Baumeler_2018}, they were  experimentally demonstrated in enhanced tabletop experiments~\cite{Procopio_2015, Rubino_2017, rubino2017experimental, Goswami_2018, Guerin_experiment, taddei2020experimental}.

The simplest example of an indefinite causal structure is the quantum-$n$-switch. Here, any permutation of the $n$ unitaries can be applied on the target system but the order in which these unitaries are applied depends on the state of an additional quantum system. For example, in the case of the quantum-2-switch, a qubit controls whether the gate $U_0$ is applied before or after another gate $U_1$. It is known that using these structures one can decide whether the two gates $U_0$ and $U_1$ commute or anticommute with a single call to each gate. Solving the same task within the standard quantum circuit model, however, requires to call at least one gate twice \cite{Chiribella_2012}. This effect has also been experimentally demonstrated by Procopio et al. \cite{Procopio_2015}. In this way, the use of indefinite causal structures allows for an advantage in the number of gates that has to be called (queries).

A generalization of this task to $n$ unitary gates, originally introduced in Araújo et al. \cite{1} and often called Fourier promise problems (FPP), can be solved with the quantum-$n$-switch and a single call to each gate. At the same time, the best known solution with a causal algorithm calls $O(n\log_2{(n)})$ gates \cite{renner2021reassessing}. This result suggests that a quantum computer with a quantum-controlled ordering of gates require asymptotically fewer resources than a quantum computer with a fixed gate ordering to solve the same task. Unfortunately, the physical conditions to achieve this advantage are very demanding: for the tasks with $n$ unitaries the dimension of the control and target systems must be at least $n!$. This makes it virtually impossible to demonstrate this computational advantage experimentally. For this reason, another generalization of the task to more unitary gates has been proposed and experimentally demonstrated (for $n=4$) by Taddei et al. \cite{taddei2020experimental}. These problems, called Hadamard promise problems (HPP), offer an advantage by using the quantum-$n$-switch compared to causal circuits as well, but most importantly require only qubits. However, so far only one task of this class with four gates is explicitly known, and it remained open whether this advantage is preserved in the limiting case of a large number of black-box gates.

Here we generalize these tasks to an arbitrary number of unitary gates and show that they (1) provide a provable gap in query complexity between a quantum-controlled ordering of gates and causal quantum circuits in the asymptotic case, and (2) require only qubit gates. In fact, while all of these tasks can be solved with the quantum-$n$-switch and a single call to each gate, we prove that a causal algorithm requires at least $2n-1$ calls to the gates. Furthermore, we show that the best known techniques with a fixed gate ordering require $O(n\log_2{(n)})$ queries and conjecture that no better causal solution exists. Our findings allow to verify experimentally the scalable computational advantage of indefinite causal structures.

\section{The Hadamard Promise Problem} \label{HPP}
In the Hadamard promise problem, originally introduced in Ref.~\cite{taddei2020experimental}, a set of $d$-dimensional unitary gates $\{U_i\}_0^{n-1}$ is given and certain permutations of these unitaries are chosen. These permutations are denoted by $\Pi_x$ where the index $x$ ranges from $0$ to $n_x-1$ and $n_x\leq n!$ is the number of selected permutations. It is promised that for some $y\in \{0,1,...,n_x-1\}$ the following relations hold:
\begin{align}
    \forall x\in \{0,1,...,n_x-1\}:  \ \Pi_x=s(x,y)\cdot \Pi_0 \, .
    \label{promise}
\end{align}
Here, the coefficients $s(x,y)$ form a $n_x \times n_x$ Hadamard matrix, an orthogonal matrix whose entries are either $+1$ or $-1$. More formally, $s(x,y)\in \{+1,-1\}$ and the rows are pairwise orthogonal to each other:\footnote{To avoid confusion, we want to mention that we label the columns with $x$ and the rows with $y$.}
\begin{align}
\begin{split}
    \forall y,y'&\in\{0,1,...,n_x-1\}:\\ &\sum^{n_x-1}_{x=0} s(x,y)\cdot s(x,y')=n_x\cdot \delta_{y, y'} \, .
\end{split}
\end{align}
The task is to find the value~$y$ for which these promises are satisfied.

The simplest HPP involves two black-box unitaries $U_0$ and $U_1$. For the two permutations  $\Pi_0=U_1U_0$ and $\Pi_1=U_0U_1$ it is promised that $\Pi_x=s(x,y)\  \Pi_0$ where $s(x,y)=(-1)^{x\cdot y}$. While the promise for $x=0$ becomes $\Pi_0=\Pi_0$, which is trivially satisfied, for $x=1$ it translates into:
\begin{align}
    U_0U_1=(-1)^{y}\cdot U_1U_0 \, .
\end{align}
Hence, the two gates either commute ($y=0$) or anticommute ($y=1$) and the task is to find out which property is the correct one. As already mentioned in the introduction, it is known that this task can be solved with the quantum-2-switch by calling each gate only once, while in any causal quantum algorithm at least one gate has to be called twice \cite{Chiribella_2012}.\\

\begin{table}[H]
\centering
\begin{tabular}{|c||c|c||cc|cc|}\hline
\multirow{2}{*}{\backslashbox{$y$}{$x$}}
&$x=0$&$x=1$&\multicolumn{3}{c}{Examples}&\\
&$(\Pi_{0}=\Pi_{0})$&$(\Pi_{1}=(-1)^{y}\ \Pi_{0})$&$U_0$&&$U_1$&\\\hline\hline
$y=0$ & 1 & 1 &$\sigma_x$&&$\sigma_x$&\\    \hline
$y=1$ & 1 & -1 &$\sigma_y$&&$\sigma_x$&\\    \hline
\end{tabular}
\caption{The Hadamard matrix for the simplest HPP in which two unitaries either commute ($y=0$) or anticommute ($y=1$). The task is to find the correct value of~$y$.}
\label{tabex1}
\end{table}

\section{Generalizing HPPs}\label{secourmethod}
For higher $n$ only a few explicit HPPs are known. In this work, we will introduce a procedure that allows us to find a HPP for any number of involved black-box gates. The main idea is that we can combine two HPPs each with $m$ and $n$ ($d$-dimensional) unitary gates into another HPP with $m+n-1$ ($d$-dimensional) unitary gates. To do so, we denote the $m_x$ permutations of the $m$ unitaries in the first HPP with $\Pi^{(1)}_{x_1}$ such that they satisfy the following promises:
\begin{align}
    \forall x_1\in \{0,1,...,m_x-1\}:  \ \Pi^{(1)}_{x_1}&=s_1(x_1, y_1)\cdot \Pi^{(1)}_{0} \, .
\end{align}
In the second HPP there are $n$ involved $d$-dimensional black-box unitaries and the $n_x$ permutations, denoted as $\Pi^{(2)}_{x_2}$, satisfy the following promises:
\begin{align}
   \forall x_2\in \{0,1,...,n_x-1\}:  \ \Pi^{(2)}_{x_2}&=s_2(x_2, y_2)\cdot \Pi^{(2)}_{0} \, .
\end{align}
Now we choose one of the $m$ unitaries from the first HPP and replace this unitary in each of the permutations $\Pi^{(1)}_{x_1}$ with $\Pi^{(2)}_{x_2}$. In this way, we obtain $n_x\cdot m_x$ new permutations that we label with $\Pi_{(x_1,x_2)}$. One can observe that these new permutations satisfy the following relations:
\begin{align}
    \Pi_{(x_1,x_2)}&=s_2(x_2,y_2)\cdot \Pi_{(x_1,0)}\\
    &=s_2(x_2,y_2)\cdot s_1(x_1,y_1)\cdot  \Pi_{(0,0)}\, .
\end{align}
Since $s_1(x_1,y_1)$ and $s_2(x_2,y_2)$ form an $m_x \times m_x$ and $n_x \times n_x$ Hadamard matrix, respectively, the resulting matrix with entries $s((x_1,x_2),(y_1,y_2)):=s_2(x_2,y_2)\cdot s_1(x_1,y_1)$ is a $(m_x\cdot n_x)\times (m_x\cdot n_x)$ Hadamard matrix. We prove this formally in Appendix~\ref{appA}. Hence, we have obtained another HPP with $m+n-1$ involved ($d$-dimensional) unitary black-box gates.

To give an example, we can consider the simplest HPP in Table~\ref{tabex1} with two involved unitaries. Let $U_0$ and $\tilde{U}_1$ be the unitaries for which it is promised that they either commute ($y_1=0$) or anticommute ($y_1=1$). The permutations $\Pi^{(1)}_{x_1}$ read then:
\begin{align}
    \Pi^{(1)}_{x_1=0}&=\tilde{U}_1U_0\\
    \Pi^{(1)}_{x_1=1}&=U_0\tilde{U}_1=(-1)^{y_1}\cdot \tilde{U}_1U_0 \, .
\end{align}
Now we can take another instance of the same HPP with $\Pi^{(2)}_{x_2=0}=U_2U_1$ and $\Pi^{(2)}_{x_2=1}=U_1U_2$ such that the two unitaries $U_1$ and $U_2$ again either commute ($y_2=0$) or anticommute ($y_2=1$):
\begin{align}
    U_1U_2=(-1)^{y_2}\ U_2U_1 \, .
\end{align}
Replacing now $\tilde{U}_1$ in both of the permutations $\Pi^{(1)}_{x_1=0}=\tilde{U}_1U_0$ and $\Pi^{(1)}_{x_1=1}=U_0\tilde{U}_1$ once with $\Pi^{(2)}_{x_2=0}=U_2U_1$ and once with $\Pi^{(2)}_{x_2=1}=U_1U_2$, we obtain in total four permutations for which the following promises hold:
\begin{align}
    \Pi_{(0,0)}=U_2U_1U_0& \, ,&\label{exam31} \\ 
    \Pi_{(0,1)}=U_1U_2U_0&=(-1)^{y_2}&U_2U_1U_0 \, ,\\
    \Pi_{(1,0)}=U_0U_2U_1&=(-1)^{y_1}&U_2U_1U_0 \, ,\\
    \Pi_{(1,1)}=U_0U_1U_2&=(-1)^{y_1+y_2}&U_2U_1U_0 \, .\label{exam34}
\end{align}
We illustrate in Table~\ref{tab3} that these relations form indeed a $4\times 4$ Hadamard matrix. In a next step, one could split one of these three unitaries into another pair of either commuting or anticommuting unitaries. In this way, one would obtain an HPP with four unitaries and eight permutations.
Following this, we obtain a HPP for every number of unitary black-box gates $n$ with $n_x=2^{n-1}$ permutations and therefore a Hadamard matrix of dimension $2^{n-1}\times 2^{n-1}$. Note, however, that we are not restricted to split a unitary into a pair of commuting or anticommuting unitaries, but replacing a unitary by any set of permutations that form a HPP by themselves is possible.

\begin{table}[H]
    \centering
\begin{tabular}{|c||c|c|c|c||c|c|cc|}\hline
\multirow{2}{*}{\backslashbox{$(y_1,y_2)$}{$(x_1,x_2)$}} &$x=$&$x=$&$x=$&$x=$& \multicolumn{3}{c}{Examples}&\\
& $(0,0)$& $(1,0)$& $(0,1)$& $(1,1)$&$U_0$& $U_1$& $U_2$& \\\hline\hline
$y=(0,0)$ & 1 & 1 & 1 & 1 &$\sigma_x$&$\sigma_x$&$\mathds{1}$& \\\hline
$y=(0,1)$ & 1 & 1 & -1 & -1 &$\sigma_x$&$\frac{\sigma_y+\sigma_z}{\sqrt{2}}$&$\frac{\sigma_y-\sigma_z}{\sqrt{2}}$& \\\hline
$y=(1,0)$ & 1 & -1 & 1 & -1 &$\sigma_y$&$\sigma_x$&$\mathds{1}$& \\\hline
$y=(1,1)$ & 1 & -1 & -1 & 1 &$\sigma_y$&$\frac{\sigma_y+\sigma_z}{\sqrt{2}}$&$\frac{\sigma_y-\sigma_z}{\sqrt{2}}$ &\\ \hline
\end{tabular}
    \caption{The Hadamard matrix for the HPP given in \eqref{exam31}-\eqref{exam34} (for short: $\Pi_{(x_1,x_2)}=(-1)^{x_1\cdot y_1+x_2\cdot y_2}\  \Pi_{(0,0)}$). For every possible combination of the parameters $y=(y_1,y_2)$ a set of unitaries that satisfy the promise is given.
    }
    \label{tab3}
\end{table}
To show that these tasks are indeed realisable, one has to prove that unitaries that satisfy these promises exist. It turns out that for many tasks of this class this can be done by a straightforward approach. For instance, we obtained the examples in Table~\ref{tab3} by simply replacing the examples of $U_1=\sigma_x$ in Table~\ref{tabex1} with a pair of unitaries that either commute (if $y_2=0$) or anticommute (if $y_2=1$) and whose product is proportional to the original unitary~$U_1=\sigma_x$:
\begin{align}
    U_1&=\sigma_x\ \xrightarrow{y_2=0}\  U_1=\sigma_x&&U_2=\mathds{1} \\
    U_1&=\sigma_x\ \xrightarrow{y_2=1}\  U_1=\frac{\sigma_y+\sigma_z}{\sqrt{2}}&&U_2=\frac{\sigma_y-\sigma_z}{\sqrt{2}}
\end{align}
In this sense, we obtain the examples for the task with $n+1$ unitaries from the examples for the task with $n$ unitaries. Since there are some subtleties with this procedure, we discuss this further in Appendix~\ref{secexistence}.

\section{Solution with the quantum-n-switch}

\begin{figure}[hbt!]
\centering
\includegraphics[width=0.5\textwidth]{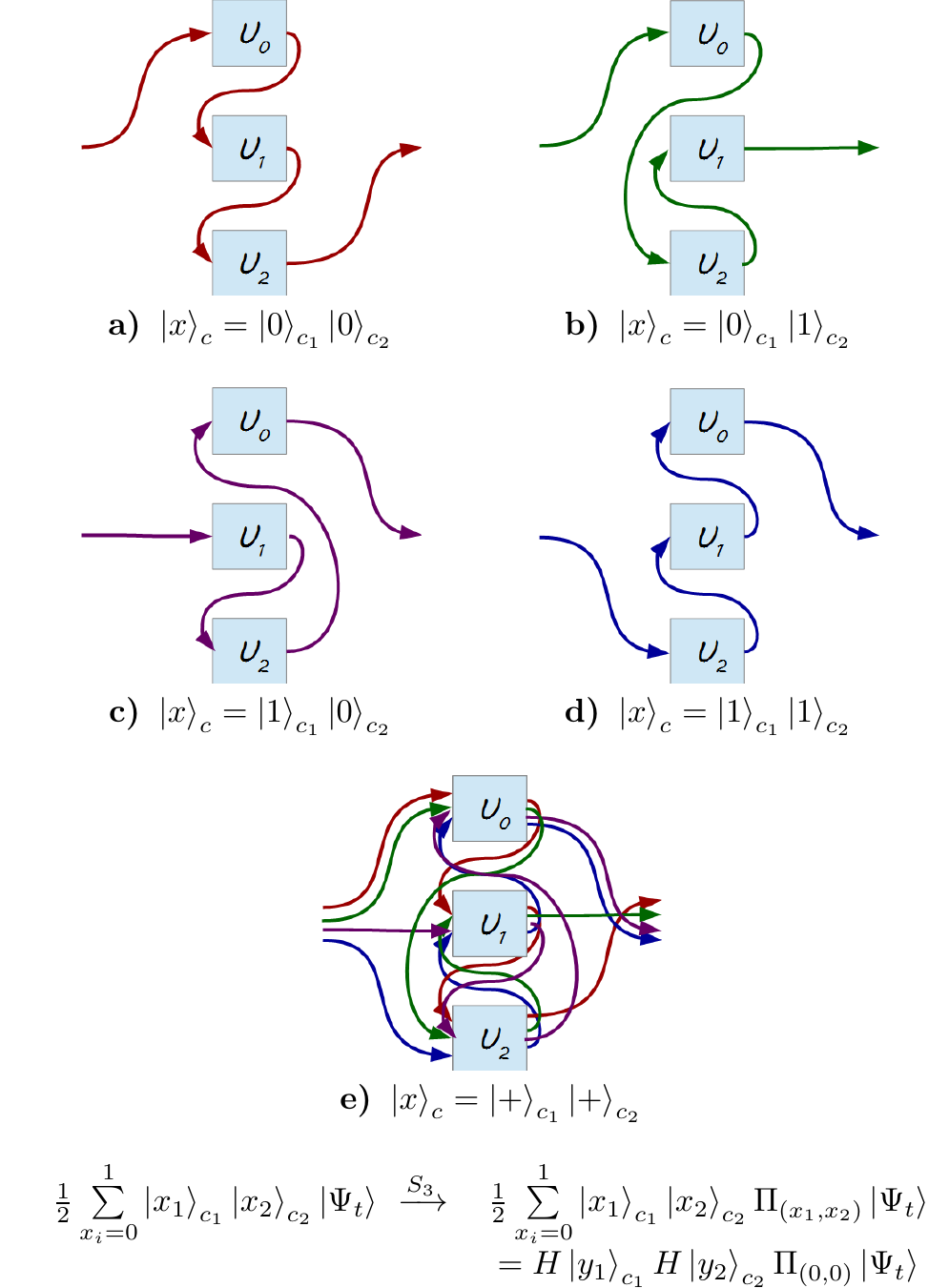}
    \caption{Solving the HPP in Table~\ref{tab3} with the 3-switch: The state of the control system $\ket{x}_c=\ket{x_1}_{c_1} \ket{x_2}_{c_2}$ determines in which order the gates are applied on the target system. If the control system is initialized in a superposition, the quantum-$3$-switch can be used to solve this HPP by calling each unitary $U_i$ only once.
    }
    \label{fig:my_label}
\end{figure}

As pointed out in Ref.~\cite{taddei2020experimental}, every HPP (independent of whether it is constructed using our method or otherwise) can be solved with the quantum-$n$-switch and a single call to each gate. The quantum-$n$-switch is denoted here as $S_n$. It is the quantum gate that applies the permutation $\Pi_x$ on the target system~$\ket{\Psi_t}$ whenever the control system is in the state $\ket{x}$:
\begin{align}
    \forall x\in\{0,1,...,n_x-1\}:\ S_n\ket{x}_c\otimes\ket{\Psi_t}=\ket{x}_c\otimes \Pi_x\ket{\Psi_t} \, .
\end{align}
Moreover, to every Hadamard matrix $s(x,y)$ we associate the corresponding unitary transformation $H_{n_x}$ that is defined as:
\begin{align}
    \forall y\in\{0,1,...,n_x-1\}:\ H_{n_x}\ket{y}=\frac{1}{\sqrt{n_x}}\sum_{x=0}^{n_x-1}s(x,y)\ket{x} \, .
\end{align}

To solve HPPs, the $n_x$-dimensional control system is first transformed into an equal superposition of all states $x\in\{0,1,...,n_x-1\}$, usually by applying a Hadamard transformation to all control qubits.
Meanwhile, the target system $\ket{\Psi_t}$ is initialized in an arbitrary $d$-dimensional state:
\begin{align}
    \left(\frac{1}{\sqrt{n_x}}\sum_{x=0}^{n_x-1}\ket{x}_c\right)\otimes\ket{\Psi_t} \, .
\end{align}
Now, if the $n$-switch is applied, depending on the state $\ket{x}$ of the control system, the permutation $\Pi_x$ is applied on the target system~$\ket{\Psi_t}$ (see Fig.~\ref{fig:my_label} for an illustration of the map for the case of $n=3$):
\begin{align}
    S_n\left(\frac{1}{\sqrt{n_x}}\sum_{x=0}^{n_x-1}\ket{x}_c\right)\otimes\ket{\Psi_t}=\frac{1}{\sqrt{n_x}}\sum_{x=0}^{n_x-1}\ket{x}_c\otimes\Pi_x\ket{\Psi_t} \, .
\end{align}
With the promise $\Pi_x=s(x,y)\cdot \Pi_0$, this state can be rewritten into:
\begin{align}
\begin{split}
    \frac{1}{\sqrt{n_x}}&\sum_{x=0}^{n_x-1}\ket{x}_c\otimes\Pi_x\ket{\Psi_t}\\
    &=\left(\frac{1}{\sqrt{n_x}}\sum_{x=0}^{n_x-1}s(x,y)\ket{x}_c\right)\otimes\Pi_0\ket{\Psi_t} \, .
    \label{switcheq1}
\end{split}
\end{align}
In this way, the target system always ends up in the state $\Pi_0\ket{\Psi_t}$ (independent of~$x$) and factorizes out. Observe that the final state of the control system is precisely $H_{n_x}\ket{y}_c$. Hence, applying the inverse (transposed) Hadamard transform $H^{-1}_{n_x}$ on the control system, we obtain:
\begin{align}
    H_{n_x}^{-1}\left(\frac{1}{\sqrt{n_x}}\sum_{x=0}^{n_x-1}s(x,y)\ket{x}_c\right)\otimes\Pi_0\ket{\Psi_t}=\ket{y}_c\otimes\Pi_0\ket{\Psi_t} \, .\label{switcheq2}
\end{align}
In this way, the solution~$y$ can be read out by a measurement of the control system in the computational basis. In the $n$-switch each unitary is called exactly once. Hence, the total query complexity of this algorithm is $n$.

\section{Solution with causal quantum algorithms}\label{secsimswitch}

\begin{figure}[hbt!]
\smaller[1]
\begin{center}
$\Qcircuit @C=0.5em @R=1em {
\lstick{\ket{0}_{c_1}}      & \gate{H} & \ctrl{3} & \qw        & \ctrl{3} & \qw       & \qw        & \qw       & \qw        & \qw      & \qw        & \qw      & \ctrlo{3} & \qw        & \ctrlo{3} & \gate{H} & \qw & \rstick{\ket{y_1}_{c_1}}             \\
\lstick{\ket{0}_{c_2}}      & \gate{H} & \qw       & \qw        & \qw       & \ctrl{3} & \qw        & \ctrl{3} & \qw        & \ctrlo{3} & \qw        & \ctrlo{3} & \qw      & \qw        & \qw      & \gate{H} & \qw & \rstick{\ket{y_2}_{c_2}}             \\
\lstick{\ket{\Psi_t}} & \qw      & \qswap    & \gate{U_0} & \qswap    & \qswap    & \gate{U_1} & \qswap    & \gate{U_2} & \qswap   & \gate{U_1} & \qswap   & \qswap   & \gate{U_0} & \qswap   & \qw      & \qw & \rstick{\Pi_0\ket{\Psi_t}} \\
\lstick{\ket{a_0}}    & \qw      & \qswap    & \qw        & \qswap    & \qw       & \qw        & \qw       & \qw        & \qw      & \qw        & \qw      & \qswap   & \qw        & \qswap   & \qw      & \qw & \rstick{U_0\ket{a_0}} \\
\lstick{\ket{a_1}}    & \qw      & \qw       & \qw        & \qw       & \qswap    & \qw        & \qswap    & \qw        & \qswap   & \qw        & \qswap   & \qw      & \qw        & \qw      & \qw      & \qw & \rstick{U_1\ket{a_1}}}$
\end{center}
\normalsize
    \caption{Simulation of the four permutations $U_2U_1U_0$, $U_1U_2U_0$, $U_0U_2U_1$, $U_0U_1U_2$ involved in the HPP given in Table~\ref{tab3} with the smallest possible number of used black-box gates. A measurement of the control qubits at the end reveals the solution $y=(y_1, y_2)$.}
    \label{soln3}
\end{figure}
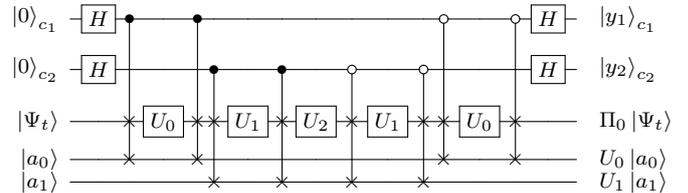

It is possible to simulate the quantum-$n$-switch with a causal algorithm and $O(n^2)$ calls to the black-box gates. Since every HPP can be solved with the quantum-$n$-switch, every simulation thereof (or more precisely the simulation of all involved permutations) can solve the same task as well. For a detailed study of the simulation of the quantum-$n$-switch we refer to Ref.~\cite{Facchini_2015} (but also Ref.~\cite{Colnaghi_2012, 1, renner2021reassessing}). For example, all permutations involved in the HPP given in Table~\ref{tab3} can be simulated with the algorithm in Fig.~\ref{soln3}. This is also the shortest possible solution since such an algorithm can be used to determine for each pair of the unitaries $U_0$, $U_1$ and $U_2$ whether the pair of unitaries commute or anticommute (by setting the remaining gate to $\mathds{1}$). Such a causal algorithm requires to call at least two of the three unitaries twice, hence at least five gates are called in total.

The idea can be extended to HPPs with a set of $n$ unitary gates. Each such problem contains as a subproblem the task of deciding for each pair of gates whether that pair commutes or anticommutes. This later problem requires a minimum number of queries and thus also determines a lower bound on the number of queries for the original problem. This is specified by the following lemma.

\begin{lemma}
Consider the class of all problems that can be generated from the HPP in Table~\ref{tabex1} with the method introduced in Section~\ref{secourmethod}. For every HPP (with a set of $n$ different gates) in that class a solution with a causal quantum algorithm has to call at least $2n-1$ unitary gates.
\end{lemma}

\begin{proof}
We can show by induction, that every solution to that task must be able to determine for every pair of unitary gates whether that pair commutes or anticommutes, when we set all remaining gates to $\mathds{1}$. For the base case of $n=2$, we note that there is only the HPP given by Table~\ref{tabex1} itself for which the statement is clearly correct.

For the induction step, remember that any task with $n+1$ unitary gates is obtained by replacing one unitary $U_i$ from a task with $n$ gates with two unitaries that we denote here as $U^{(1)}_{i}$ and $U^{(2)}_{i}$. We can check that a solution to the new task must be able to determine for every pair of gates whether that pair commutes or anticommutes: (1) If the solution for the task with $n$ gates is able to determine for every pair of unitaries whether they commute or anticommute, a solution to the new task with $n+1$ gates is able to determine for every pair $U_j$ and $U_k$ with $j,k\neq i$ this property when we set $U^{(1)}_{i}=U^{(2)}_{i}=\mathds{1}$. (2) Similar, a solution to the new task is able to determine whether $U^{(1)}_{i}$ and $U_j$ (for every $j\neq i$) commute or anticommute when we set $U^{(2)}_{i}=\mathds{1}$. The analog argument holds for $U^{(2)}_{i}$ and every $U_j$ with $j\neq i$. (3) For the remaining pair of $U^{(1)}_{i}$ and $U^{(2)}_{i}$, this follows by construction of the task since part of the solution of the new task is exactly to determine whether $U^{(1)}_{i}$ and $U^{(2)}_{i}$ commute or anticommute. This proves the induction hypothesis.

Since a causal algorithm that is able to determine whether two gates commute or anticommute has to call at least one of the two gates twice \cite{Chiribella_2012}, this requires, in total, to call at least $n-1$ gates twice. Therefore, at least $2n-1$ gates have to be called in total.
\end{proof} 

However, we believe that for most tasks in that class a causal solution has to call more than $O(n)$ gates. To motivate our conjecture, we want to point out that a similar argument as above holds for a very simple HPP that contains only two permutations and is defined by:
\begin{align}
    \Pi_0&:=U_{n-1}U_{n-2}...U_{2}U_1U_0 \, ,\\
    \Pi_1&:=U_{0}U_{1}U_{2}...U_{n-2}U_{n-1} \, .
\end{align}
It is promised that $\Pi_1=(-1)^y\ \Pi_0$ and the task is to determine $y$. A solution to that HPP is able to determine for every pair of unitaries $U_j$ and $U_k$ whether they commute or anticommute. More precisely, if we set all remaining unitaries to $\mathds{1}$, the two permutations reduce to $\Pi_0=U_kU_j$ and $\Pi_1=U_jU_k$ (given that w.l.o.g. $j<k$) from which the statement follows. In general, however, a HPP with $n$ gates that is generated with our method contains many more permutations (in fact $2^{n-1}$) and is able to determine much more structure between the unitaries.


Therefore, we conjecture that, for small $n$, a simulation of all involved permutations is the most efficient causal solution. For larger $n$, methods similar to the ones introduced in Ref.~\cite{renner2021reassessing} can be used to find more efficient solutions. Indeed, we show in Appendix~\ref{secsimswitchgeneral} that all HPPs that we can generate with our method can be solved with a causal quantum algorithm and $O(n\log_2{(n)})$ calls to the black-box gates.\footnote{There are other known techniques to solve the same tasks. They are discussed in Ref.~\cite{taddei2020experimental} and it is argued there that they require more calls to the black-box gates than a simulation of all permutations.} While we conjecture that this is the most efficient causal solution, we want to mention that there might be other problems in this class (obtained different than with our method) that offer a larger advantage.



\section{Conclusion}
Indefinite causal structures can be used to solve certain tasks more efficiently than any causally ordered quantum algorithm. In this work, we generalized a specific class of problems that provide an advantage of using a superposition of different gate orderings in the asymptotic limit. These tasks are constructed for an arbitrary number of gates and are suitable for an experimental demonstration of this computational advantage as they only involve low dimensional target systems (qubits). We showed that, while all of these tasks can be solved with the quantum-$n$-switch and a single call to each gate, causal algorithms require more calls to the black-box unitaries. We want to mention that the simplest HPP with two commuting or anticommuting unitary gates can be translated to an exponential advantage for certain communication tasks in Ref.~\cite{Guerin_2016}. We believe that our generalization of that task leads to advantages for (multipartite) communication tasks as well.


Furthermore, we found that all of these tasks can be solved with a causal algorithm and $O(n \log_2{(n)})$ calls to the black-box gates. We want to point out that currently there is no known task for which the advantage in the number of gates that has to be called is larger then $O(n)$ (for indefinite causal structures) versus $O(n \log_2{(n)})$ (for causal quantum circuits). This raises the important challenge of finding computational tasks for which indefinite causal structures provide a more significant advantage.

\section*{Acknowledgements}
We  acknowledge  financial  support  from the  Austrian  Science  Fund  (FWF)  through  BeyondC (F7103-N38), the project no. I-2906, as well as support by  the  John Templeton Foundation through grant 61466, The Quantum Information Structure of Spacetime (qiss.fr), the Foundational Questions Institute (FQXi) and the research platform TURIS. The opinions expressed in this publication are those of the authors and do not necessarily reflect the views of the John Templeton Foundation. Furthermore, we are thankful for the tutorial on Q-circuit \cite{eastin2004qcircuit} that helped a lot to create the quantum circuits in \LaTeX.

\bibliography{bib}{}
\bibliographystyle{ieeetr}

\appendix

\begin{widetext}
\section{The product of two Hadamard matrices is another Hadamard matrix}\label{appA}
\begin{lemma}
If $s_1(x_1,y_1)$ and $s_2(x_2,y_2)$ are the entries of an $m_x \times m_x$ and $n_x \times n_x$ Hadamard matrix, then $s((x_1,x_2),(y_1,y_2)):=s_2(x_2,y_2)\cdot s_1(x_1,y_1)$ forms an $(m_x\cdot n_x) \times (m_x\cdot n_x)$ Hadamard matrix.
\end{lemma}
\begin{proof}
Since $s_1(x_1,y_1)$ and $s_2(x_2,y_2)$ form Hadamard matrices, we know that $s_1(x_1,y_1),s_2(x_2,y_2)\in\{+1,-1\}$ from which we conclude that $s((x_1,x_2),(y_1,y_2))\in\{+1,-1\}$. Furthermore, since $s_1(x_1,y_1)$ and $s_2(x_2,y_2)$ form orthogonal matrices, we know:
\begin{align}
\begin{split}
    \forall y_1,y'_1\in\{0,1,...,m_x-1\}:\ \sum^{m_x-1}_{x_1=0} s_1(x_1,y_1)\cdot s_1(x_1,y'_1)=m_x\cdot \delta_{y_1, y'_1} \, .
\end{split}
\end{align}
For $s_2(x_2,y_2)$ the analog expression holds. From this we can calculate directly that $s((x_1,x_2),(y_1,y_2))$ forms an orthogonal matrix as well. In fact, two rows are orthogonal to each other:
\begin{align}
\begin{split}
    \forall (y_1,y_2), (y'_1,y'_2)\in \{0,1,...,m_x-1\}\times \{0,1,...,&n_x-1\}:\\
    \sum^{m_x-1}_{x_1=0}\sum^{n_x-1}_{x_2=0} s((x_1,x_2),(y_1,y_2))\cdot s((x_1,x_2),(y'_1, y'_2))&=\sum^{m_x-1}_{x_1=0}\sum^{n_x-1}_{x_2=0} s_1(x_1,y_1)\cdot s_2(x_2,y_2)\cdot s_1(x_1,y'_1)\cdot s_2(x_2,y'_2)\\
    &=\left(\sum^{m_x-1}_{x_1=0}s_1(x_1,y_1)\cdot s_1(x_1,y'_1)\right)\cdot\left(\sum^{n_x-1}_{x_2=0} s_2(x_2,y_2)\cdot s_2(x_2,y'_2)\right)\\
    &=m_x\cdot n_x\cdot \delta_{y_1,y'_1}\cdot \delta_{y_2,y'_2}\\
    &=(m_x\cdot n_x)\cdot \delta_{(y_1,y_2),(y'_1,y'_2)} \, .
\end{split}
\end{align}
Hence $s((x_1,x_2),(y_1,y_2))$ forms an orthogonal matrix whose entries are either $+1$ or $-1$, a Hadamard matrix.
\end{proof}

\section{Existence of unitaries that satisfy the promise}\label{secexistence}
As already mentioned in the main text, given that examples of unitaries for the task with $n$ gates exist, examples for the task with $n+1$ unitaries (in which one unitary is replaced by a pair of either commuting or anticommuting unitaries) can be found. More formally, if the unitary $U_i$ in the original HPP is of the form $U\sigma_zU^\dagger$ (for an arbitrary 2-dimensional unitary $U$) and should be replaced by a pair of commuting unitaries, one can choose for instance $U\sigma_zU^\dagger$ and $\mathds{1}$:
\begin{align}
    U\sigma_zU^\dagger&=\mathds{1}\cdot U\sigma_z U^\dagger=U\sigma_z U^\dagger \cdot \mathds{1} \, .
\end{align}
On the other hand, if $U_i$ should be replaced with two anticommuting unitaries, one can choose for example $U\sigma_x U^\dagger$ and $U(i\sigma_y) U^\dagger$ since:
\begin{align}
    U\sigma_zU^\dagger&=U\sigma_x U^\dagger \cdot U(i\sigma_y)U^\dagger=-U(i\sigma_y) U^\dagger \cdot U\sigma_xU^\dagger \, .
\end{align}
Note that such a replacement is not unique, since we can choose $U\left(\frac{\sigma_x+\sigma_y}{\sqrt{2}}\right)U^\dagger$ and $U\left(\frac{\sigma_x-\sigma_y}{\sqrt{2}}\right)U^\dagger$ as well. More precisely, replacing $U\sigma_zU^\dagger$ by $UV\sigma_x V^\dagger U^\dagger$ and $UV(i\sigma_y) V^\dagger U^\dagger$ with $V\sigma_z V^\dagger=\sigma_z$ is allowed. (Intuitively speaking, $V$ is a rotation in the $x$-$y$-plane that leaves the $z$ direction invariant.)


Nevertheless, for certain parameter combinations a problem appears with this method. Take for instance the step in which $U\sigma_zU^\dagger$ is replaced by $U\sigma_zU^\dagger$ and $\mathds{1}$. In a next step, it is impossible to replace $\mathds{1}$ by a pair of anticommuting unitaries since there are no $2\times 2$ unitaries $U_1$ and $U_2$ such that:
\begin{align}
    \mathds{1}=U_1\cdot U_2=-U_2\cdot U_1 \, .
\end{align}
(However, replacing $\mathds{1}$ by two commuting unitaries is clearly possible, for example $U\sigma_zU^\dagger$ and $U\sigma_zU^\dagger$.) Therefore, if we are only using combinations of the most simplest HPP in Table~\ref{tabex1} (replacing a unitary step by step with pairs of commuting or anticommuting unitaries), for certain parameter combinations no examples of 2-dimensional unitaries can be found. For the task itself this only implies that certain solutions $y$ are impossible. For the other parameter combinations, unitaries that satisfy the promises can still be found. We are not giving a thorough analysis of which parameter combinations are impossible since this also depends on the details of the HPP and which specific unitary $U_i$ is replaced. However, the impossibility of finding examples for certain solutions seems to be rare (especially for small $n$) and we present in the next subsection a way to circumvent this issue by changing the underlying HPP.

\subsection{Changing the underlying HPP}
\begin{table}[H]
\centering
\begin{tabular}{|c||c|c||c|c|cc|}\hline
\multirow{2}{*}{\backslashbox{$y$}{$x$}}
&$x=0$&$x=1$&\multicolumn{3}{c}{Examples}&\\
&$(\Pi_{0}=\Pi_{0})$&$(\Pi_{1}=(-1)^{y}\ \Pi_{0})$&$U_0$&$U_1$&$U_2$&\\\hline\hline
$y=0$ & 1 & 1 &$\sigma_y$&$\sigma_z$&$\sigma_z$&\\    \hline
$y=1$ & 1 & -1 &$\sigma_x$&$\sigma_y$&$\sigma_z$&\\    \hline
\end{tabular}
\caption{The Hadamard matrix for the HPP with $\Pi_0:=U_2U_1U_0$ and $\Pi_1:=U_0U_1U_2$. It is promised that $\Pi_x=(-1)^{x\cdot y}\ \Pi_0$ and the task is to find $y$. Using our method we obtain HPPs for an arbitrary (odd) number of unitary qubit gates. Note that one can also combine this HPP with the one in Table~\ref{tabex1}. For example, by replacing one of the three unitaries with a pair of commuting or anticommuting unitaries, we obtain a HPP with four unitary gates and four permutations.}
\label{tabex2}
\end{table}

\begin{figure}[hbt!]
\smaller[1]
\begin{center}
$\Qcircuit @C=0.5em @R=1em {
\lstick{\ket{0}_{c}}      & \gate{H} &  \ctrl{2}      &   \qw      & \ctrl{2}       & \ctrl{3} & \qw        & \ctrl{3} & \qw        & \ctrlo{3} & \qw        & \ctrlo{3} & \ctrlo{2}      & \qw        & \ctrlo{2}      & \gate{H} & \qw & \rstick{\ket{y}_{c}}             \\
\lstick{\ket{\Psi_t}} & \qw      & \qswap    & \gate{U_0} & \qswap    & \qswap    & \gate{U_1} & \qswap    & \gate{U_2} & \qswap   & \gate{U_1} & \qswap   & \qswap   & \gate{U_0} & \qswap   & \qw      & \qw & \rstick{\Pi_0\ket{\Psi_t}} \\
\lstick{\ket{a_0}}    & \qw      & \qswap    & \qw        & \qswap    & \qw       & \qw        & \qw       & \qw        & \qw      & \qw        & \qw      & \qswap   & \qw        & \qswap   & \qw      & \qw & \rstick{U_0\ket{a_0}} \\
\lstick{\ket{a_1}}    & \qw      & \qw       & \qw        & \qw       & \qswap    & \qw        & \qswap    & \qw        & \qswap   & \qw        & \qswap   & \qw      & \qw        & \qw      & \qw      & \qw & \rstick{U_1\ket{a_1}}}$
\end{center}
\normalsize
    \caption{Most efficient causal solution of the HPP in Table~\ref{tabex2} based on the simulation of the two permutations $\Pi_0=U_2U_1U_0$ and $\Pi_1=U_0U_1U_2$. In comparison with the algorithm in Fig.~\ref{soln3}, this algorithm uses only one control qubit instead of two which might be interesting for an experimental realisation. This causal solution is also the one with the smallest number of black-box calls since such an algorithm must be able to determine for each pair of unitary gates whether that pair commutes or anticommutes (by setting the remaining gate to $\mathds{1}$). This requires to call at least two of the three gates twice (see also Section~\ref{secsimswitch}).}
    \label{soln3b}
\end{figure}
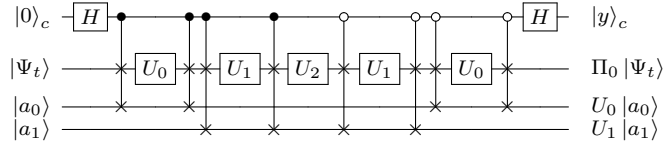

Consider the HPP with three unitaries in which $\Pi_0:= U_2U_1U_0$ and $\Pi_1:= U_0U_1U_2$ such that the two permutations satisfy the promise $\Pi_1=(-1)^y\cdot \Pi_0$ (see Table~\ref{tabex2}). We show that for every HPP that can be generated out of that HPP using our method, it is possible to find a set of 2-dimensional unitaries that satisfy the promise for every parameter combination. As one can see in Table~\ref{tabex2}, for the original HPP with three gates one can find examples of unitary gates that are only of the form $U\sigma_z U^\dagger$. Replacing a unitary of that form by three unitaries that satisfy the promise of the same HPP for $y=0$ is always possible. We can take for instance $(U \sigma_z U^\dagger)$, $(U \sigma_x U^\dagger)$ and $(U \sigma_x U^\dagger)$:
\begin{align}
    U\sigma_z U^\dagger &=(U \sigma_z U^\dagger)\cdot ( U \sigma_x U^\dagger) \cdot ( U \sigma_x U^\dagger) =(U \sigma_x U^\dagger)\cdot ( U \sigma_x U^\dagger) \cdot ( U \sigma_z U^\dagger) \, .
\end{align}
Similar if the three unitaries shall satisfy the promise for $y=1$ we can replace $U\sigma_z U^\dagger$ by $\mathds{1}$, $U\sigma_x U^\dagger$ and $U(i\sigma_y) U^\dagger$:
\begin{align}
    U\sigma_z U^\dagger &=\mathds{1}\cdot ( U \sigma_x U^\dagger) \cdot ( U (i\sigma_y) U^\dagger) =-( U (i\sigma_y) U^\dagger)\cdot ( U \sigma_x U^\dagger) \cdot \mathds{1} \, .
\end{align}
The difference is now that it is also possible to replace a unitary of the form $\mathds{1}$ into three unitaries that satisfy the promise by themselves. For $y=0$, we can take for instance $U \sigma_x U^\dagger$, $U \sigma_x U^\dagger$ and $\mathds{1}$:
\begin{align}
    \mathds{1} &= (U \sigma_x U^\dagger) \cdot (U \sigma_x U^\dagger) \cdot \mathds{1} =\mathds{1} \cdot (U \sigma_x U^\dagger) \cdot (U \sigma_x U^\dagger) \, .
\end{align}
For the case of $y=1$, we can replace $\mathds{1}$ by $U \sigma_x U^\dagger$, $U (i\sigma_y) U^\dagger$ and $U \sigma_z U^\dagger$:
\begin{align}
    \mathds{1} &= (U \sigma_x U^\dagger) \cdot (U (i\sigma_y) U^\dagger) \cdot (U \sigma_z U^\dagger) =-(U \sigma_z U^\dagger) \cdot (U (i\sigma_y) U^\dagger) \cdot (U \sigma_x U^\dagger) \, .
\end{align}
Note, however, that $U \sigma_x U^\dagger$ and $U (i\sigma_y) U^\dagger$ are again of the form $U' \sigma_z U'^\dagger$ for an appropriate choice of $U'$ and in a next step these unitaries can be replaced again with three unitaries that satisfy the promise of that HPP by themselves. Therefore, the existence of examples for all possible solutions $y$ of the HPP with $n$ unitaries implies, by induction, the existence of examples for all solutions $y$ of the resulting HPP with $n+2$ gates.


\section{Solution with causal quantum algorithms}\label{secsimswitchgeneral}
\begin{figure}[h!]
\centering
\includegraphics[width=0.97\textwidth]{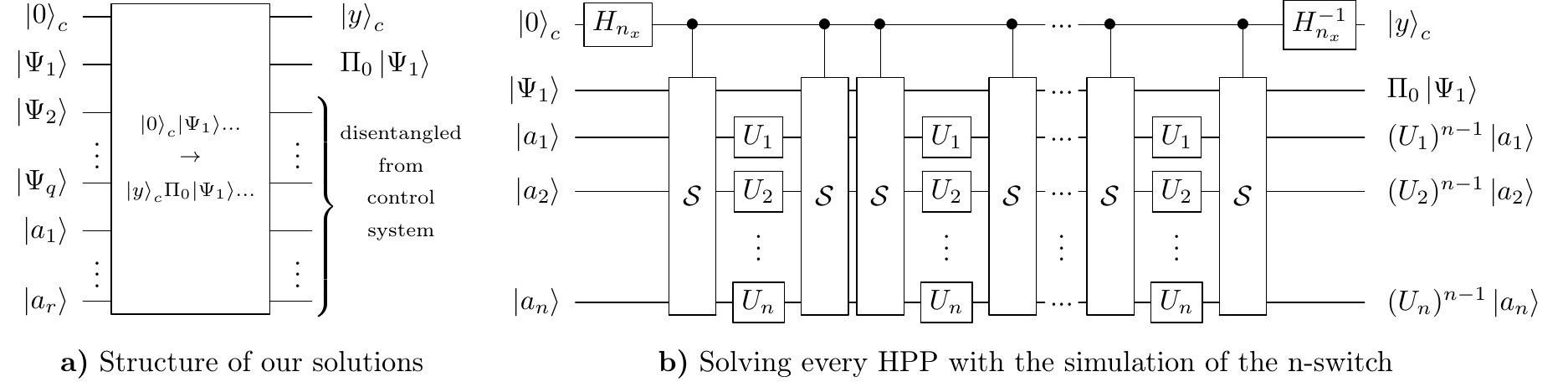}
    \caption{The structure of the algorithms we construct here is given in Fig.~\ref{figgeneral} a): All target and auxiliary systems are initialized in an arbitrary $d$-dimension state. After the algorithm is applied, the control system ends up in the state $\ket{y}_c$ from which the solution can be read out by a measurement in the computational basis. In addition, it is important for our proof that the first target system ends up in the state $\Pi_0\ket{\Psi_1}$ where $\Pi_0$ is the identity permutation of the corresponding HPP. One such algorithm is based on the simulation of the quantum-$n$-switch and given in Fig.~\ref{figgeneral} b): The Hadamard transform $H_{n_x}$ maps the initial state of the control system to an equal superposition of all considered permutations. Afterwards, the permutation $\Pi_x=U_{\sigma_x(n)}...U_{\sigma_x(2)}U_{\sigma_x(1)}$ is applied on $\ket{\Psi_1}$ by swapping the target system $\ket{\Psi_1}$ in each step $i=1,2,...,n$ with the corresponding auxiliary system $\ket{a_{\sigma_x(i)}}$. Since each auxiliary system $\ket{a_i}$ is swapped with the target system exactly once, the gate $U_i$ acts on $\ket{a_i}$ exactly $n-1$ times and ends up in the state $(U_i)^{n-1}\ket{a_i}$, independent of the state of the control system. Due to the promise $\Pi_x=s(x,y)\ \Pi_0$, the final state of the control and target system can be rewritten into $\ket{y}_c\otimes \Pi_0\ket{\Psi_1}$ (same calculation as in Eq.~\eqref{switcheq1} and Eq.~\eqref{switcheq2}) as required for the algorithm in Fig.~\ref{figgeneral} a). (To avoid confusion, we want to mention that we label the unitaries in this section (for convenience) with $1,2, ..., n$.)}
    \label{figgeneral}
\end{figure}

In this section, we will show that all HPPs that we can generate with our method from a finite set of HPPs (we call them ``fundamental'' here) can be solved with a causal quantum algorithm and $O(n\log_2{(n)})$ calls to the black-box gates. The fundamental HPP can be for example only the one given in Table~\ref{tabex1} and then the HPP given in Table~\ref{tab3} is an example of a (non-fundamental) task in that class. One can also consider the class of all tasks generated by the two fundamental HPPs given in Table~\ref{tabex1} \emph{and} Table~\ref{tabex2} which contains more tasks. Also other HPPs, not explicitly stated in this work, can be included.

\begin{theorem}
A finite set of HPPs is given and we consider the class of problems that can be generated from these fundamental HPPs with the method introduced in Section~\ref{secourmethod}. Let $k_{max}$ be the number of unitary gates contained in the fundamental HPP with the most gates and let $C:=2\cdot (k_{max}-1)$. For every problem (with a set of $n$ unitary gates) in that class there exists a causal quantum algorithm that solves this task by calling at most $C\cdot n\cdot \log_2{(n)}$ gates.
\end{theorem}
\begin{proof}
We construct for any task in that class a causal algorithm that has the form given in Fig.~\ref{figgeneral} a) and calls at most $C\cdot n\cdot \log_2{(n)}$ gates. More precisely, given that such an algorithm exists for every problem in that class with at most $n-1$ unitary gates, we construct a solution for the task with $n$ gates and the hypothesis follows by induction.\\

\textbf{Base case:}\\
For the base case, we consider all HPPs in that class that contain not more than $k_{max}$ gates (hence $2\leq n\leq k_{max}$). For these tasks, there exists a solution with a causal algorithm that calls $n^2$ gates in total. In fact, one can use the simulation of the quantum-$n$-switch given in Fig.~\ref{figgeneral} b). The hypothesis holds since $C\cdot n\cdot \log_2{(n)}=2\cdot (k_{max}-1)\cdot n\cdot \log_2{(n)}\geq 2\cdot (k_i-1)\cdot n\geq n^2$ (note that $k_{max}\geq n\geq 2$ and $\log_2{(n)}\geq 1$).\\ 

\textbf{Induction step:}\\
Consider any HPP with $n$ involved black-box gates and denote the permutations involved in this task as $\Pi_x$ (with $x\in\{0,1,..., n_x-1\}$) and the promise as $\Pi_x=s(x,y)\ \Pi_0$. Since the HPP is constructed with our method, there is a fundamental HPP from which everything starts. Let's denote the identity permutation of this HPP as $\tilde{\Pi}_0$ and let $k\leq k_{max}$ be the number of unitaries in that HPP (note that, for convenience, we label these unitaries with $1,2,..., k$ instead of $0,1,..., k-1$):
\begin{align}
    \tilde{\Pi}_0=\tilde{U}_{k}\tilde{U}_{k-1}...\tilde{U}_2\tilde{U}_1 \, .
\end{align}
The permutations $\tilde{\Pi}_{\tilde{x}}=\tilde{U}_{\tilde{\sigma}_{\tilde{x}}(k)}...\tilde{U}_{\tilde{\sigma}_{\tilde{x}}(2)}\tilde{U}_{\tilde{\sigma}_{\tilde{x}}(1)}$ of that starting HPP are permutations of the unitaries $\tilde{U}_i$ and satisfy the following relations:
\begin{align}
    \tilde{\Pi}_{\tilde{x}}=\tilde{s}(\tilde{x},\tilde{y})\ \tilde{\Pi}_0 \, .
\end{align}
Now, by applying our method, each unitary $\tilde{U}_i$ is replaced, step by step, with the permutations of other HPPs. It is important to note that these permutations form by themselves a HPP of the same class (but with less unitaries $n_i<n$):
\begin{align}
    \forall i\in \{1,2,...,k\}: \tilde{U}_i \rightarrow \Pi^{(i)}_{x_i}=s_{i}(x_i, y_i)\ \Pi^{(i)}_{0} \, .
\end{align}
Therefore, we can write the label $x$ as $(\tilde{x}, x_1, x_2, ..., x_{k})$ and the permutation $\Pi_x$ is exactly obtained by taking $\tilde{\Pi}_{\tilde{x}}$ and replacing each unitary $\tilde{U}_i$ with the corresponding permutation $\Pi^{(i)}_{x_i}$. In this way, we obtain:
\begin{align}
    \Pi_{x}=\Pi_{(\tilde{x}, x_1, x_2, ..., x_{k})}=\tilde{s}(\tilde{x},\tilde{y})\cdot \left(\prod^{k}_{i=1}s_i(x_i,y_i)\right)\ \Pi_{0,0,0,...,0} &&\implies&& s(x,y)=\tilde{s}(\tilde{x},\tilde{y})\cdot \left(\prod^{k}_{i=1}s_i(x_i,y_i)\right)\, .
\end{align}
Hence, solving the task is equivalent to find all values of $\tilde{y}, y_1, y_2, ..., y_{k-1} \text{ and } y_{k}$. Since the permutations $\Pi^{(i)}_{x_i}=s_i(x_i,y_i)\ \Pi^{(i)}_{0}$ form a HPP of the same class with $n_i\leq n-1$ involved unitaries, there is, by the induction hypothesis, for each $i$ a causal algorithm that finds $y_i$ with $C\cdot n_i\cdot \log_2{(n_i)}$ queries.

To find the remaining value $\tilde{y}$, more work is required. Let $j$ be the index of the block that contains the most unitaries ($n_j=\max\limits_{1\leq i\leq k}{\{n_i \}}$). If this index is not unique, one can choose one of them. Since $\Pi_{\tilde{x},0,0,...,0}$ is a permutation of the blocks $\Pi^{(1)}_0$, $\Pi^{(2)}_0$, ..., $\Pi^{(k-1)}_0$, $\Pi^{(k)}_0$ and $\Pi_{\tilde{x},0,0,...,0}=\tilde{s}(\tilde{x}, \tilde{y})\ \Pi_{0,0,0,...,0}$ we are able to find $\tilde{y}$, when we are able to simulate all permutations of the blocks $\Pi^{(1)}_0$, $\Pi^{(2)}_0$, ..., $\Pi^{(k)}_0$. This is achieved in the upper part of the algorithm in Figure~\ref{algmonster} by a particular simulation of the quantum-$k$-switch build out of two simulations of the quantum-$(k-1)$-switch. Here, in each step $i$, depending on the state of the control system $\ket{\tilde{x}}$, the corresponding block $\Pi^{(\tilde{\sigma}_{\tilde{x}}(i))}_0$ (with $\tilde{\sigma}_{\tilde{x}}(i)\neq j$) is applied on the target system $\ket{\Psi^{(j)}_1}$ by swapping $\ket{\Psi^{(j)}_1}$ with the corresponding auxiliary system $\ket{\tilde{a}_{\tilde{\sigma}_{\tilde{x}}(i)}}$. At the point where the block $\Pi^{(j)}_0$ shall be applied on $\ket{\Psi^{(j)}_1}$, the algorithm in the middle is used to realize this transformation. The first part requires at most $k-1$ steps since there are at most $k-1$ blocks in $\Pi_{\tilde{x},0,0,...,0}$ before $\Pi^{(j)}_0$. Afterwards the same procedure is used to simulate all blocks $\Pi^{(i)}_0$ that appear after $\Pi^{(j)}_0$, which requires again at most $k-1$ steps. Since each auxiliary system $\ket{\tilde{a}_i}$ is swapped exactly once, it ends up in the state $(\Pi^{(i)}_0)^{2(k-1)-1}\ket{\tilde{a}_i}$, independent of the state of the control system $\ket{\tilde{x}}$.


In total, this algorithm consumes
\begin{align}
    \sum_{i=1}^{k} C\cdot n_i\cdot \log_2{n_i}+\sum_{i=1, i\neq j}^{k} 2\cdot (k-1) \cdot n_i\leq C\cdot n\cdot \log_2{n}
\end{align}
queries. The first term corresponds to the algorithms that find all values of $y_i$ and the second term comes from the simulation of all permutations of the blocks $\Pi^{(1)}_0$, $\Pi^{(2)}_0$, ..., $\Pi^{(k)}_0$ (each block contains $n_i$ black-box gates and appears $2\cdot (k-1)$ times, except $\Pi^{(j)}_0$ which does not appear at all). To see that this expression is smaller than $C\cdot n\cdot \log_2{n}$, we observe that:
\begin{align}
\begin{split}
    &C\cdot n\cdot \log_2{n}-\sum_{i=1}^{k} C\cdot n_i\cdot \log_2{n_i}=C\cdot \left(\sum_{i=1}^{k} n_i\right)\cdot \log_2{n}-C\cdot \sum_{i=1}^{k}  n_i\cdot \log_2{n_i}\\
    &=C\cdot \left(\sum_{i=1}^{k} n_i\cdot \log_2{\frac{n}{n_i}}\right)
    \geq C\cdot \left(\sum_{i=1, i\neq j}^{k} n_i\cdot \log_2{\frac{n}{n_i}}\right)\geq C\cdot \left(\sum_{i=1, i\neq j}^{k} n_i\cdot \log_2{2}\right)\\
    &= 2\cdot (k_{max}-1)\cdot \left(\sum_{i=1, i\neq j}^{k} n_i\right)\geq\sum_{i=1, i\neq j}^{k} 2\cdot (k-1) \cdot n_i \, .
\end{split}
\end{align}
Here, we have used that $\frac{n}{n_i}\geq 2$ for every $i\neq j$ (since $n_j\geq n_i$ and $n\geq n_j+n_i$). This concludes the proof.
\end{proof}



\begin{figure}
\begin{center}
\includegraphics[width=0.97\textwidth]{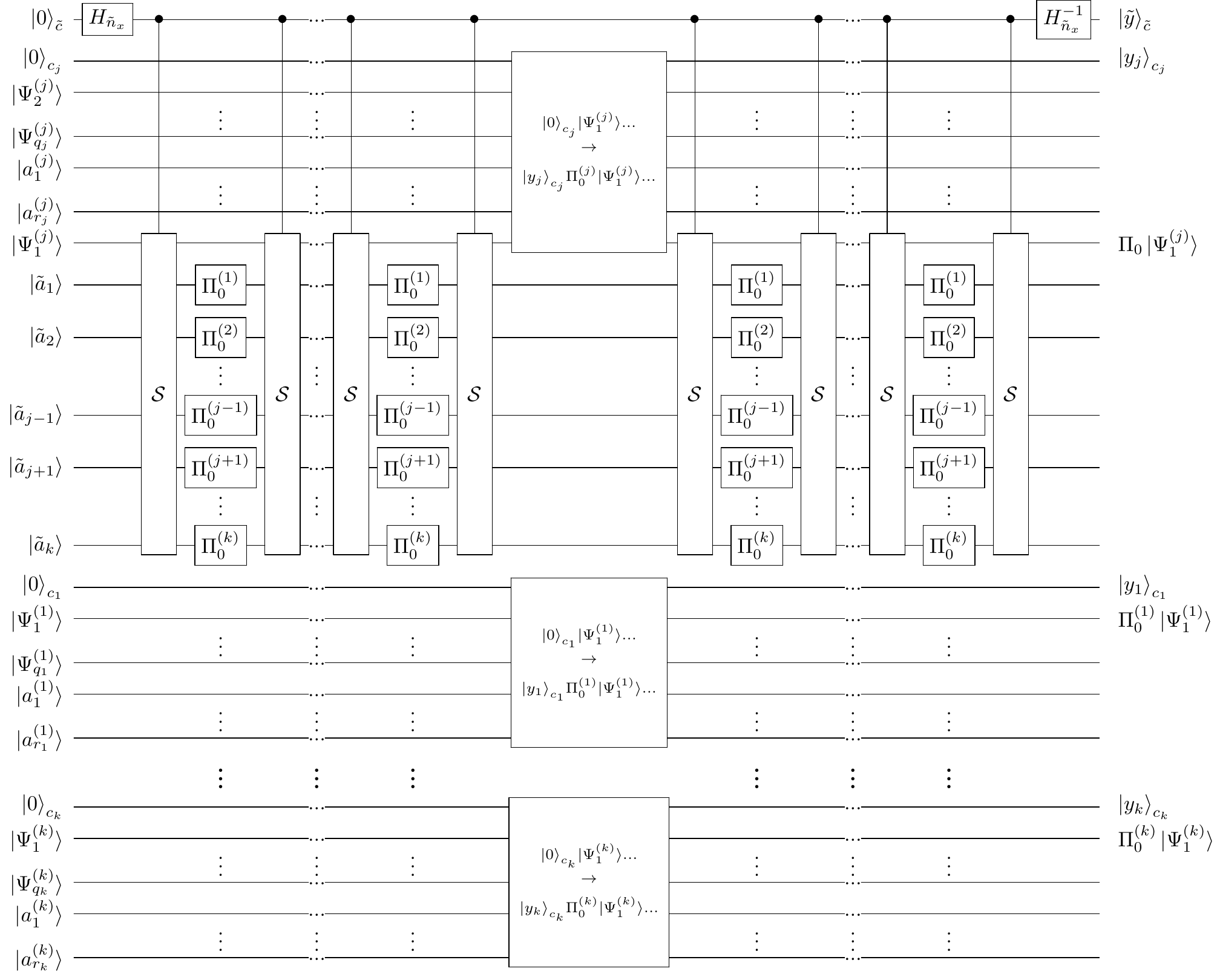}
\end{center}
    \caption{The algorithm that finds $y=(\tilde{y}, y_1, y_2, ..., y_k)$: The values $y_i$ for $i\neq j$ are found in the lower part of the algorithm completely independent of the rest (the index $j$ is skipped in the lower part). In the upper part, we simulate all possible permutations of the blocks $\Pi^{(1)}_0$, $\Pi^{(2)}_0$, ..., $\Pi^{(k)}_0$ (which is sufficient to determine $\tilde{y}$) by a construction that simulates the quantum-$k$-switch with two simulations of the quantum-$(k-1)$-switch. Here, in each step $i$ and depending on the state of the control system $\ket{\tilde{x}}$ the corresponding block $\Pi^{(\tilde{\sigma}_{\tilde{x}}(i))}_0$ (with $\tilde{\sigma}_{\tilde{x}}(i)\neq j$) is applied on $\ket{\Psi^{(j)}_1}$ by swapping that target system with $\ket{\tilde{a}_{\tilde{\sigma}_{\tilde{x}}(i)}}$. The block $\Pi^{(j)}_0$ is applied in the middle step. In this way, the target system $\ket{\Psi^{(j)}_1}$ ends up in the state $\Pi_0\ket{\Psi^{(j)}_1}$ and the solution $\tilde{y}$ can be read out in the control system.}
    \label{algmonster}
\end{figure}

\end{widetext}

\end{document}